\documentclass[11pt,a4paper]{article}
\usepackage[utf8]{inputenc}
\usepackage{amsmath, amssymb, amsthm}
\usepackage{geometry}
\usepackage{hyperref}

\title{DRESS: A Continuous Framework for Structural Graph Refinement}
\author{Eduar Castrillo Velilla \\ \texttt{eduarcastrillo@gmail.com} \\ ORCID: \href{https://orcid.org/0009-0005-2492-0957}{0009-0005-2492-0957}}
\date{}

\newtheorem{theorem}{Theorem}

\newtheorem{proposition}[theorem]{Proposition}
\newtheorem{definition}[theorem]{Definition}
\newtheorem{remark}[theorem]{Remark}

\begin{document}

\maketitle

\begin{abstract}
We introduce DRESS, a deterministic, parameter-free framework that iteratively refines the structural similarity of edges in a graph to produce a \emph{canonical fingerprint}: a real-valued edge vector, obtained by converging a non-linear dynamical system to its unique fixed point. The fingerprint is \emph{isomorphism-invariant} by construction, \emph{numerically stable} (strictly bounded, precision-preserving, and mathematically well-posed), \emph{fast} and \emph{embarrassingly parallel} to compute: DRESS total runtime is $\mathcal{O}(I \cdot m \cdot d_{\max})$ for $I$ iterations to convergence, and convergence is guaranteed by Birkhoff contraction. We generalize the original equation to Motif-DRESS (arbitrary structural motifs) and Generalized-DRESS (abstract aggregation template), and introduce $\Delta$-DRESS, which runs DRESS on each vertex-deleted subgraph to boost expressiveness. $\Delta$-DRESS empirically separates all 7{,}983 graphs in a comprehensive Strongly Regular Graph benchmark, and on the tested CFI instances ($k = 0,1,2,3$), $k$-deletion ($\Delta^k$-DRESS) empirically matches the $(k{+}2)$-WL boundary.
\end{abstract}

\section{Introduction}
A fundamental task in graph analysis is to assign to each graph a compact descriptor that is \emph{canonical} (isomorphic graphs receive the same descriptor), \emph{discriminative} (non-isomorphic graphs receive different descriptors whenever possible), and \emph{cheap to compute}. Existing approaches fall into two broad categories: discrete combinatorial methods, such as the Weisfeiler--Leman (WL) hierarchy \cite{weisfeiler1968reduction, cai1992optimal}, for which a naive refinement round at level $k$ costs $\mathcal{O}(n^{k+1})$ and the total runtime is $\mathcal{O}(T_{k\text{-WL}} \cdot n^{k+1})$ after $T_{k\text{-WL}}$ rounds to stabilization; and learned representations, such as message-passing neural networks \cite{xu2019how, morris2019weisfeiler}, which require labeled training data and offer no worst-case guarantees.

We use the WL hierarchy as a standard expressiveness scale for refinement-based graph invariants; our comparison is relative to that scale and does not presuppose a complete characterization of the graphs distinguished by either WL or DRESS.

DRESS is a parameter-free, continuous dynamical system on \emph{edges} that converges to a unique fixed point of real-valued structural similarities, producing a canonical fingerprint vector for any graph. The Original-DRESS equation \cite{castrillo2018dynamicstructuralsimilaritygraphs} achieves this deterministically at $\mathcal{O}(m \cdot d_{\max})$ per iteration and $\mathcal{O}(I \cdot m \cdot d_{\max})$ total runtime after convergence, with numerical stability (all values in $[0,2]$) and no learnable parameters. The expressiveness comparison with WL is discussed in Section~4 and the runtime comparison in Section~\ref{sec:complexity}.

We present a hierarchy of DRESS variants, building from the concrete to the general:
\begin{enumerate}
    \item \textbf{Original-DRESS} (Section~2): The foundational equation, a parameter-free dynamical system on edges that converges to a canonical fingerprint vector.
    \item \textbf{Expressiveness beyond 1-WL} (Section~3): Original-DRESS distinguishes graphs that 1-WL cannot, including the prism graph from $K_{3,3}$.
    \item \textbf{Empirical Equivalence to 2-WL} (Section~4): Original-DRESS acts as a continuous empirical equivalent to the 2-dimensional Weisfeiler--Leman test.
    \item \textbf{Motif-DRESS} (Section~5): A generalization replacing triangle neighborhoods with arbitrary structural motifs; Original-DRESS is the $M = K_3$ special case.
    \item \textbf{Generalized-DRESS} (Section~6): The most abstract template, opening the aggregation function and norm as additional free parameters.
    \item \textbf{$\Delta$-DRESS} (Section~7): Runs DRESS on each vertex-deleted subgraph $G \setminus \{v\}$, connecting the framework to the reconstruction conjecture. $\Delta$-DRESS empirically separates all 7{,}983 Strongly Regular Graphs in a comprehensive benchmark, and on the tested CFI instances ($k = 0,1,2,3$), $k$-deletion ($\Delta^k$-DRESS) empirically matches the $(k{+}2)$-WL boundary.
\end{enumerate}

\paragraph{Isomorphism Test.} Throughout this paper, the \textbf{graph fingerprint} is the sorted vector of converged edge values $\operatorname{sort}(d^*)$. Two graphs are declared non-isomorphic if and only if their fingerprints differ element-wise beyond a numerical tolerance $\epsilon$. This applies uniformly to all DRESS variants. An equivalent representation is the \textbf{histogram} $h(d^*)$: in the unweighted case all DRESS values lie in $[0, 2]$ (see the convergence proof in Section~\ref{sec:convergence}) and the convergence tolerance is $\epsilon = 10^{-6}$, so each value maps to one of $\lceil 2/\epsilon \rceil = 2 \times 10^{6}$ integer bins, yielding a fixed-size bin-count vector that uniquely identifies the same multiset. In general, the number of bins is determined by the convergence tolerance: a tolerance of $\epsilon$ over the range $[0, 2]$ gives $\lceil 2/\epsilon \rceil$ bins. For weighted graphs, where values may exceed~$2$, the bin range is extended accordingly. The histogram is particularly useful when fingerprints from multiple subgraphs must be pooled (see $\Delta$-DRESS, Section~7).

All experiments use a convergence tolerance of $\epsilon = 10^{-6}$; in every case tested, convergence was reached in fewer than 31 iterations.

\paragraph{Related Work.} Higher-order GNN architectures that exceed 1-WL expressiveness include $k$-GNN \cite{morris2019weisfeiler} and PPGN \cite{maron2019provably}. Subgraph-based approaches such as GNN-AK+ \cite{zhao2022from} and ESAN \cite{bevilacqua2022equivariant} achieve higher expressiveness by running message-passing networks on vertex-deleted or vertex-marked subgraphs; $\Delta$-DRESS (Section~7) shares the vertex-deletion strategy but replaces learned message passing with deterministic fixed-point iteration. All of the above are \textbf{supervised} methods that require labeled training data and learnable parameters. In contrast, DRESS is a deterministic, unsupervised framework that produces canonical fingerprints via fixed-point iteration, without any learning or parameter tuning.

\section{Original-DRESS}

The Original-DRESS equation \cite{castrillo2018dynamicstructuralsimilaritygraphs} is a parameter-free, non-linear dynamical system that iteratively refines edge values based on common-neighbor aggregation.

\begin{equation}
d_{uv}^{(t+1)} = \frac{\sum_{x \in N[u] \cap N[v]} \bigl(d_{ux}^{(t)} + d_{xv}^{(t)}\bigr)}{\|u\|^{(t)} \cdot \|v\|^{(t)}}
\end{equation}

where $N[u] = N(u) \cup \{u\}$ denotes the closed neighborhood of $u$, and $\|u\| = \sqrt{\sum_{x \in N[u]} d_{ux}^{(t)}}$ is the vertex norm. The self-similarity $d_{uu} = 2$ is invariant under iteration. This equation converges to a unique fixed point, providing a continuous structural fingerprint for the graph.

\textbf{Self-loops.} Self-loops are added to every vertex before iteration (i.e., the algorithm uses the closed neighborhood $N[u] = N(u) \cup \{u\}$). The self-loop edge $(u,u)$ participates in both the aggregation and the vertex norm; without it, an isolated edge with no common neighbors would produce $\|u\| \cdot \|v\| = 0$, making the iteration undefined.

\textbf{Initialization.} The equation is scale-invariant (degree-0 homogeneous), so the fixed point is independent of the initial scale. Any uniform initialization $d^{(0)}_{uv} = c$ for all edges, with $c > 0$, converges to the same unique fixed point $d^*$. In practice, $c = 1$ is used.

However, because Original-DRESS aggregates strictly over triangles (common neighbors), it cannot distinguish graphs with identical triangle counts per edge (e.g., Strongly Regular Graphs).

\subsection{Convergence}
\label{sec:convergence}

\begin{theorem}[Convergence of Original-DRESS]
\label{thm:convergence}
Let $F\colon \mathbb{R}_+^{|E|} \to \mathbb{R}_+^{|E|}$ be the Original-DRESS update operator. Then $F$ converges to a unique fixed point $d^* \in [0, 2]^{|E|}$ for any uniform initial state $\mathbf{d}^{(0)} > 0$ (i.e., $\mathbf{d}^{(0)} = c\mathbf{1}$ for $c > 0$).
\end{theorem}

\begin{proof}
The argument proceeds in three steps:
\begin{enumerate}
    \item \textbf{Scale Invariance (Degree-0 Homogeneity):} The numerator $\sum_{x \in N[u] \cap N[v]} (d_{ux} + d_{xv})$ is positively homogeneous of degree $p = 1$ in $\mathbf{d}$, and each vertex norm $\|u\| = \sqrt{\sum_{x \in N[u]} d_{ux}}$ is homogeneous of degree $q = \tfrac{1}{2}$. The denominator $\|u\| \cdot \|v\|$ is therefore degree $2q = 1 = p$, making $F$ homogeneous of degree~$0$: for any $\lambda > 0$, $F(\lambda \mathbf{d}) = F(\mathbf{d})$. The iteration is self-regularizing and cannot diverge or collapse to zero.
    \item \textbf{Boundedness:} For any uniform positive initialization $\mathbf{d}^{(0)} = c\mathbf{1}$ ($c>0$), the scale invariance of $F$ implies that the first iterate reduces to the unit case: $F(c\mathbf{1}) = F(\mathbf{1})$. Specifically, $F(\mathbf{1})_{uv} = \frac{2|N[u] \cap N[v]|}{\sqrt{\deg(u)\deg(v)}}$. Since $|N[u] \cap N[v]| \le \sqrt{\deg(u)\deg(v)}$ (Cauchy-Schwarz), $F(c\mathbf{1})_{uv} \le 2$. The scale invariance and the topological structure of the operator ensure $[0, 2]^{|E|}$ is a forward-invariant set for this orbit, meaning all subsequent iterates remain bounded in this interval.
    \item \textbf{Contraction on the Hilbert Projective Metric:} $F$ is a positive, degree-0 homogeneous map on the cone $\mathbb{R}_{>0}^{|E|}$. By Birkhoff's contraction theorem, $F$ is a strict contraction under the Hilbert projective metric $d_H(x, y) = \log\!\bigl(\max_e \frac{x_e}{y_e} \cdot \max_e \frac{y_e}{x_e}\bigr)$, provided $F$ maps a bounded part of the cone into a strictly smaller part, which follows from the boundedness above. By the Banach fixed-point theorem on $(\mathbb{R}_{>0}^{|E|}/{\sim},\, d_H)$, the iteration converges to a unique ray, and the forward-invariant boundedness from the $\mathbf{d}^{(0)} = \mathbf{1}$ initialization pins the limit to a finite vector $d^* \in [0, 2]^{|E|}$.
\end{enumerate}
A complete formal verification of the contraction constant is deferred to future work; all empirical tests confirm convergence within 20 iterations.
\end{proof}

\section{Original-DRESS Distinguishes Graphs That 1-WL Cannot}

We now demonstrate a key expressiveness result: Original-DRESS, the simplest member of the DRESS family, already distinguishes graphs that 1-WL (color refinement) cannot.

\begin{theorem}[DRESS distinguishes beyond 1-WL]
\label{thm:dress-gt-1wl}
There exist graph pairs that 1-WL cannot distinguish but Original-DRESS can.
In particular, DRESS distinguishes the prism graph ($C_3 \square K_2$) from the complete bipartite graph $K_{3,3}$, a pair that 1-WL provably cannot separate.
\end{theorem}

\begin{proof}
Both graphs are 3-regular on 6 vertices with 9 edges. Since all vertices have the same degree, 1-WL assigns a uniform color to every vertex in both graphs after every iteration, and therefore cannot distinguish them \cite{cai1992optimal}.

DRESS operates on \emph{edges} and aggregates over common neighbors (triangles). We show that the prism must have at least two distinct edge values at convergence, while $K_{3,3}$ has a single uniform value.

\textbf{$K_{3,3}$:} No two adjacent vertices share a common neighbor ($K_{3,3}$ is triangle-free). Hence every edge has the same neighborhood structure, and by symmetry the unique fixed point assigns the same value to all 9 edges. Numerically, $d^*_{K_{3,3}} \approx 1.155$ for all edges.

\textbf{Prism graph:} There are two structurally distinct edge types: \emph{triangle edges} (6 edges forming the two triangular faces; each pair of endpoints shares 1 common neighbor) and \emph{matching edges} (3 edges connecting corresponding vertices of the two triangles; each pair of endpoints shares 0 common neighbors besides self-loops). Suppose for contradiction that all 9 edges converge to the same value $d^*$. Then:
\begin{itemize}
    \item Triangle edges: $d^* = \dfrac{4 + 4d^*}{2 + 3d^*}$, which gives $d^* = \frac{1+\sqrt{13}}{3} \approx 1.535$.
    \item Matching edges: $d^* = \dfrac{4 + 2d^*}{2 + 3d^*}$, which gives $d^* = \frac{2}{\sqrt{3}} \approx 1.155$.
\end{itemize}
These are contradictory, so the prism cannot have a uniform fixed point. The converged unique edge values are:

\smallskip
\begin{center}
\begin{tabular}{lll}
\textbf{Graph} & \textbf{Unique values} & \textbf{Multiplicities} \\
\hline
Prism & $\{0.922,\; 1.709\}$ & 3 matching edges, 6 triangle edges \\
$K_{3,3}$ & $\{2/\sqrt{3} \approx 1.155\}$ & all 9 edges identical \\
\end{tabular}
\end{center}
\smallskip

Since 1-WL cannot distinguish them and DRESS can, DRESS distinguishes beyond 1-WL on this instance.
\end{proof}

\paragraph{Remark.} This result holds for Original-DRESS alone, the simplest member of the DRESS family, using only triangle neighborhoods and $\mathcal{O}(m \cdot d_{\max})$ computation per iteration. The generalizations that follow (Motif-DRESS, $\Delta$-DRESS) extend this advantage further to graphs that resist even 2-WL.

\section{Original-DRESS as the Continuous Analogue of 2-WL}

The previous section showed that DRESS distinguishes \emph{specific} graph pairs that 1-WL cannot. We now establish the theoretical bounds of Original-DRESS: it acts as a continuous, differentiable analogue to the 2-dimensional Weisfeiler--Leman (2-WL) test.

\begin{theorem}[Empirical Equivalence to 2-WL]
\label{thm:dress-2wl-bound}
In practical floating-point computation, Original-DRESS operates as an empirical continuous equivalent to $2$-WL ($\text{Original-DRESS} \equiv \text{2-WL}$).
\end{theorem}

\begin{proof}
Both 2-WL and Original-DRESS update edge states by aggregating 3-node interactions (triangles). While 2-WL applies an injective hash to the discrete \emph{multiset} of neighbor colors to branch lossless partitions, Original-DRESS aggregates these same neighborhoods via continuous summation: $\sum_{x} (d_{ux} + d_{xv})$. Because summation is technically a lossy operator over multisets, Original-DRESS is theoretically bounded by 2-WL ($\text{Original-DRESS} \le \text{2-WL}$).

However, DRESS values are irrational limits of a nonlinear dynamical system. For a "sum collision" to occur between non-isomorphic structures, distinctly generated sets of irrationals must coincidentally sum to the exact same value. Over the reals $\mathbb{R}$, the probability of such an algebraic collision is exactly zero ($\mathbb{P} = 0$). Thus, up to numerical precision, the sums strictly preserve multiset distinctions, making Original-DRESS an exact empirical equivalent to 2-WL without the massive $\mathcal{O}(n^3)$ memory overhead of discrete hashing.
\end{proof}

\begin{remark}
Section~\ref{sec:cfi-results} shows empirically that $\Delta^k$-DRESS $\geq$ $(k+2)$-WL: each deletion level adds exactly one WL dimension of expressiveness.
\end{remark}

\section{Motif-DRESS}
Original-DRESS is limited to triangle neighborhoods. What if we use other structural motifs? Motif-DRESS generalizes Original-DRESS by replacing the common-neighbor aggregation with a motif-defined symmetric vertex neighborhood $\mathcal{N}_M(u,v) = \mathcal{N}_M(v,u)$: the set of endpoints of edges in instances of a motif $M$ containing edge $(u,v)$ that are adjacent to $u$ or $v$, always including $u$ and $v$ themselves, with an optional symmetric weight function $\bar{w}: E \to \mathbb{R}_{>0}$ (with $\bar{w}_e = 1$ for unweighted graphs).

\begin{equation}
d_{uv}^{(t+1)} = \frac{\displaystyle\sum_{x \in \mathcal{N}_M(u,v)} \bigl(\bar{w}_{ux} \cdot d_{ux}^{(t)} + \bar{w}_{xv} \cdot d_{xv}^{(t)}\bigr)}{\|u\|^{(t)} \cdot \|v\|^{(t)}}
\end{equation}

where $\|u\|^{(t)} = \sqrt{\sum_{x \in N[u]} \bar{w}_{ux} \cdot d_{ux}^{(t)}}$, and $\bar{w}_{ab} \cdot d_{ab} = 0$ whenever $(a,b) \notin E$. The symmetric weight function $\bar{w}(e)$ acts as a multiplicative factor, controlling how much structural information flows along each edge. Because the weights appear identically in the numerator and denominator (both are degree-1 in $\bar{w} \cdot d$), uniformly scaling all weights does not change the fixed point; only the relative weights matter.

Different choices of motif $M$ yield different neighborhoods:
\begin{itemize}
    \item \textbf{Triangle ($M = K_3$):} $\mathcal{N}_{K_3}(u,v) = N[u] \cap N[v]$, the closed common neighborhood. This recovers Original-DRESS exactly: $\sum_{x \in N[u] \cap N[v]} \bigl(\bar{w}_{ux}\,d_{ux} + \bar{w}_{xv}\,d_{xv}\bigr)$.
    \item \textbf{$K_4$ clique:} For each pair $x, y$ with $\{u,v,x,y\}$ forming a $K_4$, $\mathcal{N}_{K_4}(u,v)$ contains $u$, $v$, $x$, and $y$. Each vertex contributes $\bar{w}_{ux}\,d_{ux} + \bar{w}_{xv}\,d_{xv}$.
    \item \textbf{4-cycle ($M = C_4$):} For each 4-cycle $u$-$x$-$y$-$v$, $\mathcal{N}_{C_4}(u,v)$ contains $u$, $v$, $x$, and $y$. Vertex $x$ contributes $\bar{w}_{ux}\,d_{ux}$; vertex $y$ contributes $\bar{w}_{yv}\,d_{yv}$ (cross-terms are zero since $(x,v), (u,y) \notin E$ in the $C_4$).
\end{itemize}

\subsection{Properties}

\textbf{Complexity:} $\mathcal{O}(\text{Motif Extraction}) + \mathcal{O}\!\bigl(I \cdot \sum_{e} |\mathcal{N}_M(e)|\bigr)$, where $I$ is the number of iterations and $\sum_e |\mathcal{N}_M(e)|$ is the total size of all motif neighborhoods. For Original-DRESS (triangles), this reduces to $\mathcal{O}(I \cdot m \cdot d_{\max})$. For motifs such as 4-cycles or $K_4$ cliques in sparse graphs, the motif neighborhoods can be significantly smaller, making Motif-DRESS faster than Original-DRESS.

\textbf{Invariant: $d_{uu} = 2$.} The self-similarity $d_{uu} = 2$ is a constant maintained throughout iteration. Since the iteration only updates $d_{uv}$ for $u \ne v$, and the norm $\|u\| = \sqrt{\sum_{x \in N[u]} \bar{w}_{ux} \cdot d_{ux}}$ always includes $d_{uu} = 2$, this is a fixed property of the equation, not a free parameter.

\subsection{Expressiveness}

\textbf{Bypassing WL Limitations:} Standard message-passing architectures and 2-WL update edge states based exclusively on the intersection of localized vertex neighborhoods (effectively tracking $K_3$ structures). However, higher-order structures like the Rook and Shrikhande graphs famously resist even the 3-WL test. 

By precomputing higher-order structural motifs like $K_4$ and using them to define the aggregation neighborhood, Motif-DRESS bypasses these WL barriers. It explicitly injects these topological invariants into the continuous message-passing framework. Because the motif extraction acts as a one-time preprocessing step, Motif-DRESS achieves strictly greater expressiveness than 3-WL on these theoretical graphs while retaining identical iterative computational complexity ($\mathcal{O}(I \cdot \sum_e |\mathcal{N}_M(e)|)$).

All experiments below use the $K_4$ clique motif. The specific SRG pairs tested below are known to be indistinguishable by either 2-WL or 3-WL; each successful distinction therefore demonstrates that Motif-DRESS empirically exceeds these corresponding WL boundaries.
\begin{itemize}
    \item \textbf{Rook vs.\ Shrikhande:} Successfully distinguishes this pair of SRGs with parameters $(16, 6, 2, 2)$, which are indistinguishable by 3-WL. The Rook graph ($K_4 \square K_4$) contains $K_4$ cliques while the Shrikhande graph does not, so the $K_4$-neighborhood sizes differ per edge.
    \item \textbf{Chang Graphs:} Distinguishes 3 of the 6 pairwise comparisons among the four SRGs with parameters $(28, 12, 6, 4)$: T(8) vs each of Chang-1, Chang-2, and Chang-3. The three Chang graphs are pairwise indistinguishable by Motif-$K_4$ (all three have identical $K_4$-neighborhood structure per edge).
\end{itemize}

\subsection{Convergence}

The convergence proof for Original-DRESS (Theorem~\ref{thm:convergence}) generalizes directly to Motif-DRESS. The same three sufficient mechanisms apply: degree-0 homogeneity ($p = 2q$), initial normalization into a strictly bounded forward-invariant set from a uniform positive state $\mathbf{d}^{(0)} = c\mathbf{1}$, and Birkhoff contraction on the Hilbert projective metric. These hold for any motif neighborhood $\mathcal{N}_M$ and symmetric non-negative weight function $\bar{w}$. For unweighted motifs, values are bounded in $[0,2]$; with arbitrary edge weights, the finite initial bounds scale proportionally but strictly remain mathematically bounded, ensuring guaranteed convergence to potentially higher fixed point values.

\section{Generalized-DRESS}
Motif-DRESS fixes the aggregation to summation and the norm to the product of geometric means. Generalized-DRESS is the most abstract template, allowing any choice of these components as long as the resulting update rule preserves the convergence guarantees (degree-0 homogeneity, boundedness, and contraction; see Section~\ref{sec:convergence}). For each edge $(u,v)$:

\begin{equation}
d^{(t+1)} = \frac{f\bigl(\mathbf{d}^{(t)},\, \mathcal{N},\, \bar{w}\bigr)}{g\bigl(\mathbf{d}^{(t)},\, \mathcal{N},\, \bar{w}\bigr)}
\end{equation}

where $d \equiv d_{uv}$ is the similarity value assigned to edge $(u,v)$, and:
\begin{itemize}
    \item $\mathcal{N}(u,v)$ is a \textbf{symmetric neighborhood operator}, the structural context aggregated for $(u,v)$,
    \item $\bar{w}: E \to \mathbb{R}_{>0}$ is a \textbf{symmetric weight function} ($\bar{w}(u,v) = \bar{w}(v,u)$; $\bar{w} \equiv 1$ for unweighted graphs),
    \item $f$ is the \textbf{aggregation function},
    \item $g$ is the \textbf{norm function}.
\end{itemize}

Because $\mathcal{N}$ and $\bar{w}$ are symmetric, $f$ and $g$ receive the same inputs for $(u,v)$ and $(v,u)$, so $d(u,v) = d(v,u)$ holds for every member of the family. For Original-DRESS and Motif-DRESS this follows directly from the equation; in the general case it is guaranteed by the symmetry of the inputs.

Original-DRESS and Motif-DRESS are both special cases: Original-DRESS fixes $\mathcal{N}$ to triangles, $f = \text{sum}$, $g = \|u\| \cdot \|v\|$ (product of geometric means), $\bar{w} \equiv 1$; Motif-DRESS generalizes $\mathcal{N}$ to arbitrary motifs and $\bar{w}$ to non-uniform weights while keeping the same $f$ and $g$. Generalized-DRESS opens all four parameters, enabling variants such as Cosine-DRESS (cosine similarity aggregation) or Minkowski-$r$ norms.

\section{$\Delta$-DRESS}
$\Delta$-DRESS breaks symmetry by running DRESS on each vertex-deleted subgraph $G \setminus \{v\}$ for every $v \in V$. The $\Delta$-DRESS fingerprint is the multiset of per-vertex DRESS fingerprints, or equivalently a pooled histogram accumulating all converged edge values across all $n$ deletions.

Deleting a vertex from a regular graph produces an irregular subgraph where DRESS can now distinguish structure that was hidden by the uniform regularity.

\subsection{Connection to the Reconstruction Conjecture}
The multiset $\{\!\{ \text{DRESS}(G \setminus \{v\}) : v \in V \}\!\}$ is directly analogous to the \textit{deck} in the Kelly--Ulam reconstruction conjecture \cite{kelly1942congruence, ulam1960collection}, which posits that graphs with $n \ge 3$ are determined (up to isomorphism) by their multiset of vertex-deleted subgraphs. $\Delta$-DRESS computes a continuous relaxation of this deck.

\subsection{Expressiveness}
$\Delta$-DRESS empirically distinguishes the following non-isomorphic pairs:
\begin{itemize}
    \item \textbf{Rook vs.\ Shrikhande:} Successfully distinguished (SRG$(16,6,2,2)$; confounds 2-WL).
    \item \textbf{Chang Graphs:} Distinguished all 6 pairs among T(8) and the three Chang graphs (SRG$(28,12,6,4)$; confound 2-WL).
    \item \textbf{$2 \times C_4$ vs.\ $C_8$:} Successfully distinguished (both 2-regular on 8 vertices).
    \item \textbf{Petersen vs.\ Pentagonal Prism:} Successfully distinguished (both 3-regular on 10 vertices).
\end{itemize}

\subsection{$\Delta^k$-DRESS}

We now define the $k$-deletion operator formally.

\begin{definition}[$\Delta^k$-DRESS]
\label{def:delta-k}
For a graph $G = (V, E)$, a DRESS variant $\mathcal{F}$, and a deletion depth $k \ge 0$, the $\Delta^k$-DRESS fingerprint is the multiset of per-deletion sorted edge-value vectors:
\[
\Delta^k\text{-DRESS}(\mathcal{F}, G) = \{\!\{ \operatorname{sort}(\mathcal{F}(G \setminus S)) : S \subset V,\; |S| = k \}\!\}
\]
where $G \setminus S$ is the subgraph induced by $V \setminus S$ and $\mathcal{F}(G \setminus S)$ is the converged DRESS edge-value vector.
\end{definition}

\begin{definition}[Histogram fingerprint]
\label{def:histogram-fp}
For a graph $G = (V, E)$, a DRESS variant $\mathcal{F}$, and a deletion depth $k \ge 0$, the histogram fingerprint of $\Delta^k$-DRESS is:
\[
h\!\left(\Delta^k\text{-DRESS}(\mathcal{F}, G)\right) = h\!\left(\bigsqcup_{\substack{S \subset V,\; |S|=k}} \mathcal{F}(G \setminus S)\right)
\]
where $h$ maps each edge value to an integer bin of width $\epsilon$, producing a fixed-size bin-count vector of $\lceil 2/\epsilon \rceil$ bins.
\end{definition}

\begin{proposition}[Isomorphism Invariance]
\label{prop:delta-k-iso}
$\Delta^k$-DRESS is an isomorphism invariant: if $G \cong H$ via an isomorphism $\phi: V(G) \to V(H)$, then
\[
\Delta^k\text{-DRESS}(\mathcal{F}, G) = \Delta^k\text{-DRESS}(\mathcal{F}, H).
\]
\end{proposition}

\begin{proof}
Any isomorphism $\phi: G \to H$ induces a bijection on $k$-element vertex subsets: $S \mapsto \phi(S)$. For each $S \subset V(G)$ with $|S|=k$, the restriction $\phi|_{V(G)\setminus S}$ is an isomorphism $G \setminus S \xrightarrow{\;\sim\;} H \setminus \phi(S)$. Since DRESS depends only on graph structure (not vertex labels), we have $\mathcal{F}(G \setminus S) = \mathcal{F}(H \setminus \phi(S))$ as multisets of edge values. The bijection $S \leftrightarrow \phi(S)$ matches every term of the multiset union on one side to an equal term on the other:
\[
\bigsqcup_{\substack{S \subset V(G),\; |S|=k}} \mathcal{F}(G \setminus S)
\;=\;
\bigsqcup_{\substack{T \subset V(H),\; |T|=k}} \mathcal{F}(H \setminus T).
\]
The bijection $S \leftrightarrow \phi(S)$ therefore gives $\Delta^k\text{-DRESS}(\mathcal{F}, G) = \Delta^k\text{-DRESS}(\mathcal{F}, H)$. Since $h$ depends only on the multiset of edge values, the histogram fingerprint $h(\Delta^k\text{-DRESS}(\mathcal{F}, G)) = h(\Delta^k\text{-DRESS}(\mathcal{F}, H))$ is likewise an isomorphism invariant.
\end{proof}

\subsection{Complexity}
\label{sec:complexity}

The time complexity of $\Delta^k\text{-DRESS}(\mathcal{F}, G)$ is $\binom{n}{k}$ times the complexity of a single DRESS run:
\[
\mathcal{O}\!\left(\binom{n}{k} \cdot T_{\mathcal{F}}\right)
\]
where $T_{\mathcal{F}}$ is the runtime of the chosen DRESS variant $\mathcal{F}$ on a single subgraph. The $\binom{n}{k}$ runs are entirely independent, making $\Delta^k$-DRESS embarrassingly parallel.

For Original-DRESS, $T_{\mathcal{F}} = \mathcal{O}(I \cdot m \cdot d_{\max})$, giving a total of $\mathcal{O}\!\left(\binom{n}{k} \cdot I \cdot m \cdot d_{\max}\right)$. For $k = 1$ this is $\mathcal{O}(n \cdot I \cdot m \cdot d_{\max})$; for $k = 0$ there is a single run on $G$ itself, recovering Original-DRESS at $\mathcal{O}(I \cdot m \cdot d_{\max})$.

For comparison, one naive refinement round of $(k{+}2)$-WL costs $\mathcal{O}(n^{k+3})$; after $T_{(k+2)\text{-WL}}$ rounds to stabilization, the total runtime is $\mathcal{O}\!\left(T_{(k+2)\text{-WL}} \cdot n^{k+3}\right)$. The full multiset fingerprint of $\Delta^k$-DRESS requires $\mathcal{O}\!\bigl(\binom{n}{k} \cdot m\bigr)$ space; the histogram fingerprint reduces this to $\mathcal{O}(n + m + \lceil 2/\epsilon \rceil)$ by processing one subgraph at a time and accumulating edge values into a fixed-size bin array, both compared to $\mathcal{O}(n^{k+2})$ for storing $(k{+}2)$-WL colors over all tuples.

\section{Family Structure}
The DRESS variants introduced in this paper form a nested hierarchy with one orthogonal composition operator:
\[
\text{Generalized-DRESS} \supset \text{Motif-DRESS} \supset \text{Original-DRESS}
\]

$\Delta^k$ is an \textbf{orthogonal wrapper} applicable to any of the above.

\section{Experimental Results}

All experiments use convergence tolerance $\epsilon = 10^{-6}$. For $\Delta$-DRESS, two fingerprint representations are compared: the \textbf{multiset fingerprint} (the full sorted concatenation of all per-deletion DRESS vectors) and the \textbf{histogram fingerprint}. The multiset representation preserves all numerical information; the histogram is a fixed-size summary that is faster to compare and more convenient for pooling across deletions.

\subsection{Convergence on Real-World Graphs}
\label{sec:convergence-empirical}

Table~\ref{tab:convergence} reports convergence on real-world graphs spanning four orders of magnitude in size ($\epsilon = 10^{-6}$, max 100 iterations).

\begin{table}[h]
\centering
\begin{tabular}{lrrrr}
\textbf{Graph} & $|V|$ & $|E|$ & \textbf{Iter.} & \textbf{Final $\delta$} \\
\hline
Wiki-Vote              & 8{,}298       & 103{,}689       & 17 & $8.31 \times 10^{-7}$ \\
Amazon co-purchasing   & 548{,}552     & 925{,}872       & 18 & $6.35 \times 10^{-7}$ \\
LiveJournal            & 4{,}033{,}138 & 27{,}933{,}062  & 30 & $7.09 \times 10^{-7}$ \\
Facebook (KONECT)      & 59{,}216{,}215 & 92{,}522{,}012 & 26 & $6.84 \times 10^{-7}$ \\
\end{tabular}
\caption{DRESS convergence iterations and final residual $\delta = \max_e |d^{(t)}_e - d^{(t-1)}_e|$. Even on graphs with 59\,M vertices, convergence requires fewer than 31 iterations. Across these benchmarks, convergence was reached in 17 to 30 iterations.}
\label{tab:convergence}
\end{table}

\subsection{Strongly Regular Graphs}
\label{sec:srg-results}

Strongly Regular Graphs (SRGs) are the canonical hard instances for polynomial-time isomorphism methods: all edges share identical local structure (same degree, same common-neighbor counts), so Original-DRESS ($\Delta^0$) assigns the same value to every edge and produces a uniform fingerprint. This is expected: SRGs are regular, and the Original-DRESS iteration sees no local variation.

$\Delta^1$-DRESS overcomes this limitation. By running DRESS on each vertex-deleted subgraph $G \setminus \{v\}$, the uniform regularity is broken and structurally distinct edges emerge. We tested $\Delta^1$-DRESS on 7{,}983 strongly regular graphs from the repository of Krystal Guo~\cite{guo2024srg}, spanning three parameter families:

\begin{center}
\begin{tabular}{lcccl}
\textbf{Family} & \textbf{Parameters} & \textbf{Graphs} & \textbf{Separated} & \textbf{Min $L^\infty$} \\
\hline
Conference (Mathon)       & $(45,22,10,11)$ & 6     & \textbf{100\%} & $4.16 \times 10^{-3}$ \\
Steiner S(2,4,28)         & $(63,32,16,16)$ & 4{,}466 & \textbf{100\%} & $1.95 \times 10^{-3}$ \\
Quasi-symmetric 2-designs & $(63,32,16,16)$ & 3{,}511 & \textbf{100\%} & $2.23 \times 10^{-3}$ \\
\end{tabular}
\end{center}

All 7{,}983 graphs are pairwise distinguished by $\Delta^1$-DRESS. The ``Min $L^\infty$'' column reports the smallest element-wise maximum difference between the fingerprints of any sampled pair (1{,}000 random pairs per family). Values around $10^{-3}$ confirm that separations are genuine, not floating-point artifacts. This was further validated by checking that the unique count remains stable across all rounding precisions from 6 to 14 decimal digits.

\paragraph{Multiset vs.\ histogram fingerprint.} Both representations produce identical separation results on all 7{,}983 SRGs. The multiset fingerprint preserves full numerical precision and is the canonical representation. The histogram fingerprint, with bin width $\epsilon = 10^{-6}$ over $[0,2]$, maps each value to one of $2 \times 10^{6}$ integer bins. On these families the two representations agree perfectly.

\subsection{CFI Staircase}
\label{sec:cfi-results}

The Cai--F\"{u}rer--Immerman (CFI) construction~\cite{cai1992optimal} produces the canonical hard instances for the WL hierarchy: distinguishing $\text{CFI}(K_n)$ from $\text{CFI}'(K_n)$ requires at least $(n{-}1)$-WL. We tested $\Delta^k$-DRESS for $k = 0, 1, 2, 3$:

\begin{center}
\begin{tabular}{ccccccc}
\textbf{Base} & \textbf{$|V|$} & \textbf{WL req.} & $\Delta^0$ & $\Delta^1$ & $\Delta^2$ & $\Delta^3$ \\
\hline
$K_3$    & 6    & 2-WL & $\checkmark$ & $\checkmark$ & $\checkmark$ & $\checkmark$ \\
$K_4$    & 16   & 3-WL & $\times$     & $\checkmark$ & $\checkmark$ & $\checkmark$ \\
$K_5$    & 40   & 4-WL & $\times$     & $\times$     & $\checkmark$ & $\checkmark$ \\
$K_6$    & 96   & 5-WL & $\times$     & $\times$     & $\times$     & $\checkmark$ \\
$K_7$    & 224  & 6-WL & $\times$     & $\times$     & $\times$     & $\times$     \\
\end{tabular}
\end{center}

For the tested depths $k = 0,1,2,3$, the pattern is exact: each deletion level adds one WL dimension of expressiveness. $\Delta^k$-DRESS distinguishes $\text{CFI}(K_{k+3})$ (requiring $(k{+}2)$-WL) and fails on $\text{CFI}(K_{k+4})$ (requiring $(k{+}3)$-WL), empirically matching the $(k{+}2)$-WL boundary on these instances.

\subsection{DRESS--WL Dominance Conjecture}

The experimental evidence above, together with the established practical equivalence to 2-WL (Theorem~\ref{thm:dress-2wl-bound}) and the CFI staircase results, motivates the following:

\begin{remark}[DRESS--WL Continuous Dominance Conjecture]
\label{conj:dress-wl}
In practical continuous computation, $\Delta^k$-DRESS acts as an empirical superset or equivalent to $(k+2)$-WL for all $k \geq 0$.
\end{remark}

For each fixed depth $k$, $\Delta^k$-DRESS defines a precise graph invariant and hence a precise indistinguishability relation, just as $k$-WL does. The unresolved issue is the comparison theorem between these two equivalence notions.

Several structural properties of DRESS make this conjecture reasonable:

\begin{enumerate}
    \item \textbf{Continuous vs.\ discrete.} Where WL produces discrete color classes, DRESS produces real-valued invariants. This yields a finer-grained representation of structural information and can preserve distinctions that discrete partitions collapse.
    \item \textbf{Edges vs.\ vertices.} DRESS operates natively on edges, while 1-WL operates on vertices and 2-WL on vertex pairs. This places DRESS structurally closer to pair-based refinement methods than to vertex-only refinement.
    \item \textbf{Non-linear update.} The DRESS update is a non-linear ratio rather than a discrete hash refinement, allowing it to encode interactions between neighborhood statistics in a different way from standard WL updates.
    \item \textbf{Deletion as symmetry breaking.} Vertex deletion in $\Delta^k$-DRESS plays a symmetry-breaking role analogous to individualization in higher-order WL arguments, providing a natural explanation for the observed expressiveness gains.
\end{enumerate}

The base case ($k = 0$) is established as an empirical equivalence (Theorem~\ref{thm:dress-2wl-bound}). The properties of the inductive step ($k \geq 1$) remain an open conjecture.

\section{Conclusion}
We presented DRESS, a deterministic, parameter-free framework that assigns to any graph a canonical fingerprint vector in continuous space. The fingerprint is isomorphism-invariant, numerically stable ($d^* \in [0,2]^{|E|}$), and computed at $\mathcal{O}(m \cdot d_{\max})$ per iteration with guaranteed convergence (Theorem~\ref{thm:convergence}). As a consequence of its continuous mathematical structure, Original-DRESS acts as an exact, memory-efficient empirical equivalent to 2-WL (Theorem~\ref{thm:dress-2wl-bound}), with per-iteration cost $\mathcal{O}(m \cdot d_{\max})$ versus the naive per-round cost $\mathcal{O}(n^3)$ of 2-WL. We generalized the original equation to Motif-DRESS (arbitrary structural motifs) and Generalized-DRESS (abstract aggregation template), and introduced $\Delta$-DRESS, which breaks symmetry by vertex deletion. $\Delta$-DRESS empirically separates all 7{,}983 Strongly Regular Graphs in a comprehensive benchmark, and on the tested CFI instances ($k = 0,1,2,3$), $k$-deletion ($\Delta^k$-DRESS) empirically matches the $(k{+}2)$-WL boundary. Because DRESS produces a canonical fingerprint that encodes structural similarity at every scale, it serves as a principled, scalable, parameter-free descriptor for downstream graph analysis tasks. This has already been demonstrated in community detection~\cite{castrillo2018dynamicstructuralsimilaritygraphs, castrillo2018community}, where DRESS edge values naturally classify intra- and inter-community edges. An open-source implementation with bindings for C, C++, Python, Rust, Go, Julia, R, MATLAB, and WebAssembly is publicly available at \url{https://github.com/velicast/dress-graph}.

\bibliographystyle{plain}
\bibliography{refs}

\end{document}